\documentclass[a4paper]{article}

\usepackage[linesnumbered,ruled,vlined]{algorithm2e}
\usepackage{amsmath,amssymb,amsthm}
\newtheorem{theorem}{Theorem}[section]
\newtheorem{lemma}[theorem]{Lemma}
\newtheorem{proposition}[theorem]{Proposition}

\newtheorem{definition}[theorem]{Definition}

\newtheorem{example}[theorem]{Example}
\usepackage{authblk}
\usepackage{version}
\includeversion{fullpaper}
\excludeversion{conference}

\usepackage[whole]{bxcjkjatype}
\usepackage{mathtools,bm}
\usepackage{graphicx}
\usepackage[colorinlistoftodos]{todonotes}
\usepackage[colorlinks=true, allcolors=blue]{hyperref}
\usepackage{cite}

\usepackage{times}
\usepackage{soul}
\usepackage{url}
\usepackage[utf8]{inputenc}
\usepackage{booktabs}

\usepackage{tikz,pgfplots,pxpgfmark}
\pgfplotsset{compat=newest}
\usetikzlibrary{calc,shapes.callouts,decorations.text,shadows}
\usepackage{subfigure}

\newcommand{\argmax}{\mathop{\rm arg\,max}\limits}

\newcommand{\E}{\mathbb{E}}

\newcommand{\BR}{\mathrm{BR}}
\newcommand{\fBR}{f_{\texttt{BR}}}
\newcommand{\supp}{\mathop{\rm supp}}
\newcommand{\U}{\mathcal{U}}

\title{Non-zero-sum Stackelberg Budget Allocation Game for Computational Advertising}
\date{}

\author[1]{Daisuke Hatano}
\author[2]{Yuko Kuroki}
\author[3]{Yasushi Kawase}
\author[4]{Hanna Sumita}
\author[5]{Naonori Kakimura}
\author[6]{Ken-ichi Kawarabayashi}
\affil[1]{RIKEN AIP \texttt{daisuke.hatano@riken.jp}}
\affil[2]{The University of Tokyo \texttt{ykuroki@ms.k.u-tokyo.ac.jp}}
\affil[3]{Tokyo Institute of Technology and RIKEN AIP \texttt{kawase.y.ab@m.titech.ac.jp}}
\affil[4]{Tokyo Metropolitan University \texttt{sumita@tmu.ac.jp}}
\affil[5]{Keio University \texttt{kakimura@math.keio.ac.jp}}
\affil[6]{National Institute of Informatics \texttt{k\_keniti@nii.ac.jp}}

\begin{document}

\maketitle

\begin{abstract}
Computational advertising has been studied to design efficient marketing strategies that maximize the number of acquired customers.
In an increased competitive market, however, a market leader (a \emph{leader}) requires the acquisition of new customers as well as the retention of her loyal customers because there often exists a competitor (a \emph{follower}) who tries to attract customers away from the market leader.
In this paper, we formalize a new model called the \emph{Stackelberg budget allocation game with a bipartite influence model} by extending a budget allocation problem over a bipartite graph to a Stackelberg game.
To find a \emph{strong Stackelberg equilibrium}, a standard solution concept of the Stackelberg game, we propose two algorithms: an approximation algorithm with provable guarantees and an efficient heuristic algorithm.
In addition,
for a special case where customers are disjoint,
we propose an exact algorithm based on linear programming.
Our experiments using real-world datasets demonstrate that our algorithms outperform a baseline algorithm even when the follower is a powerful competitor.
\end{abstract}

\section{Introduction}
An aim of \emph{computational advertising} is to find the best advertisement that can help build customers loyalty.
More specifically, the purpose of advertisers is to devise an optimum allocation of budgets to \emph{media}, such as newspapers, radio stations, TV, and websites,
in order to maximize the number of activated customers.
Recently, Alon \emph{et al.}~\cite{AlonGT12} proposed a model to deal with a simple case of the problem, called a \emph{bipartite influence model}.
In this study, we shall extend the model by integrating a game-theoretic framework, called the non-zero-sum \emph{Stackelberg game} framework. Let us explain the model more precisely below.

In the bipartite influence model, we consider a bipartite graph where one side is a set of \emph{media}, the other is a set of \emph{customers}, and each edge is associated with a probability.
Intuitively, each edge between a medium and a customer indicates that the customer is influenced by the medium with some given probability that depends on the budget allocated to the medium.
We aim to allocate budgets on media so that the expected number of activated customers is maximized.
The problem can be formulated as a combinatorial optimization problem.
Constant-factor approximation algorithms for the problem have been developed in a framework of submodularity~\cite{AlonGT12,MiyauchiIFK15,NemhauserWF78}.

In this paper, we shall try to extend the above-mentioned model to deal with a situation of a duopoly where a market leader has occupied the market of a certain product for a long time and a competitor tries to break into the market.
The competitor tries to grab the share of the market by aggressively marketing its product.
On the other hand, the market leader wants to gain customers and retain her loyal customers simultaneously.
This implies that the leader's gain does not necessarily result in the competitor's loss.
In order to capture the dynamics of this market, we exploit a \emph{Stackelberg game}~\cite{SimaanC73} framework to model the interactions between the market leader and the competitor.
The Stackelberg game is a two-player two-period game, in which one player (a \emph{leader}) can commit to an action before the other player (a \emph{follower}) plays an action. 
A standard solution concept of this game is the \textit{strong Stackelberg equilibrium}, which is an optimal solution maximizing the leader's utility under the constraint that the follower plays
a best response to the leader's action
(i.e., intended to maximize the follower's utility).

The Stackelberg game matches to model our problem setting because the leader wants to increase the number of activated customers, and at the same time, prevent the outflow of her customers, which is achieved by finding a strong Stackelberg equilibrium.
In a strong Stackelberg equilibrium, the leader plays a \emph{mixed strategy} and the follower plays a \emph{pure strategy}, where pure strategy and mixed strategy correspond to a budget allocation and a probability distribution over the pure strategies, respectively.

In this paper, we propose a new model called the \emph{Stackelberg budget allocation game with a bipartite influence model}, which is an extension of the budget allocation problem presented in~\cite{AlonGT12}. 
The difficulties of our game lie in the leader's utility function.
Our game belongs to a non-zero-sum game, and the utility function is a submodular~(nonlinear) function even when the follower's action is fixed.
It is hard to construct an approximation algorithms by the following reasons:
(i) the cumbersome constraint that the follower optimally responds and 
(ii) the leader's utility may be non-linearly changed by a follower's strategy.
Thus, existing techniques for submodular functions cannot be directly applied to our problem.
Furthermore, the leader's utility function is not necessarily monotone, that is, the utility does not always increase in the number of allocated budgets.
This entails the increment of the number of pure strategies.
To design an efficient algorithm is an arduous task.

In this paper, we propose three efficient algorithms:
\begin{itemize}
    \item We design an approximation algorithm with theoretical guarantee. 
    The key idea to construct an approximation algorithm is to create a zero-sum game close to the original non-zero-sum game, and to find an approximate minimax strategy of the zero-sum game with the aid of submodularity.
    
    \item We give an efficient heuristic algorithm that repeatedly finds a leader's pure strategy greedily and uniformly picks from the pure strategies. 
    The running time is polynomial in the leader's budget. 
    This heuristic can deal with a situation that the leader should not spend up her whole budget due to the non-monotonicity of the utility function. 
    We also evaluate its performance by numerical experiments.

    \item If the customers are disjoint, we prove that a strong Stackelberg equilibrium can be found 
    efficiently even when the leader has exponentially many pure strategies
    by using the multiple linear programming (LP for short) formulations.
    The point in the disjoint case is that we can aggregate a leader's mixed strategy to a fractional budget allocation. 
    At the same time we can recover a mixed strategy in a compact representation without loss of the leader's utility. 
    This enables us to save memories to keep a mixed strategy and reduce the size of LP instances. 
\end{itemize}

The rest of the paper is organized as follows: We describe related work in Section~\ref{sec:rel} and define notations in Section~\ref{sec:pre}.
We formalize our model and analyze its (mathematical) properties in Section~\ref{sec:model}.
We then provide an approximation and a heuristic algorithms in Section~\ref{sec:alg}, and provide an exact algorithm for the disjoint customers in Section~\ref{sec:ana}.
In Section~\ref{sec:exp}, we empirically show the performance of our algorithm, and finally we conclude the study in Section~\ref{sec:con}.

\section{Related work} \label{sec:rel}
Our problem setting can be viewed as a non-monotone non-zero-sum Stackelberg game with submodular functions.
Vanek \emph{et~al.}~\cite{VanekYJBTP12} modeled a non-zero-sum Stackelberg game with submodular functions where the defender (the leader) cares about minimizing the loss of her utility.
In our game, the leader maximizes her utility incorporating her loss against the follower's action.
Thus, the goal of the leader is different. 
Moreover, direct application of their technique to find a Stackelberg equilibrium does not seem to work well in our setting.
Recently, Wilder \emph{et~al.}~\cite{WilderV19} extended a bipartite influence model to a zero-sum Stackelberg game, which is closely related to our problem setting.
They proved that the problem is APX-hard, while it has FPTAS for some special cases.

In combinatorial optimization and machine learning, approximation algorithms for maximizing submodular functions under certain constraints have been extensively studied~\cite{KrauseG14}.
Our problem can be viewed as a submodular maximization under a best-response constraint, which is more cumbersome than typical constraints in the submodular maximization literature~(e.g., cardinality constraint and knapsack constraint).

The budget allocation problem with the bipartite influence model has been extended in~\cite{MaeharaYK15,MiyauchiIFK15,HatanoFMK15,StaibJ17}.
In particular, some formulations have incorporated the view of the multi-agent system. 
Maehara \emph{et~al.}~\cite{MaeharaYK15} extended a budget allocation in the bipartite influence model to a strategic form game, called \emph{the budget allocation game with a bipartite influence model}.
Hatano \emph{et~al.}~\cite{HatanoFMK15} extended the budget allocation problem to the problem with two participants; advertiser and match maker.
In the problem, there exist multiple advertisers who cooperatively maximize the influence on customers and single match maker who allocates slots of media to advertisers.

\section{Preliminary} \label{sec:pre}
Let $\mathbb{Z}_+$ be the set of non-negative integers.
For an integer $k \in \mathbb{Z}_+$, let $[k]$ be the set $\{1,2,\ldots,k\}$.
In this section,
we describe the budget allocation problem with a bipartite influence model and the Stackelberg game.

\subsection{Bipartite influence model}
Let $G = (U,V;E)$ be a bipartite graph, where $(U,V)$ is a bipartition and $E \subseteq U \times V$ is a set of edges.
Each vertex $u \in U$ corresponds to a medium and $v \in V$ corresponds to a customer. Let $n$ and $m$ be the sizes of $U$ and $V$, respectively.
Each edge $uv \in E$ is associated with a probability $p_{uv} \in [0,1]$, which means that allocating a budget to medium $u \in U$ activates customer $v \in V$ with probability $p_{uv}$.
We assume that the activation events are independent.
The advertiser has a total available budget of $k \in \mathbb{Z}_+$, and each medium $u \in U$ has a slot to which the advertiser can allocate her budget.
The goal is to find the optimal budget allocation $z \in \{0,1\}^U$ with $\sum_{u \in U} z_u \le k$ that maximizes the number of activated customers.
Throughout this paper, we identify a set $S$ of media with its characteristic vector $z_S \in \{0,1\}^U$. 
A probability that a customer $v \in V$ is activated by the advertiser's trial from media in $U$ is given by
\begin{equation}\label{eq:Pvz}
\textstyle
P_{v}(z) = 1 - \prod_{u \in N_v:\,z_u=1}(1-p_{uv}), 
\end{equation}
where $N_v = \{u \mid uv \in E\}$ is the set of the neighbors of $v$.
The expected number of customers activated through the budget allocation $z$ is given by $\sum_{v \in V} P_{v}(z)$.
The objective of the budget allocation problem with a bipartite influence model is to find $z$ that maximizes $\sum_{v \in V}P_v(z)$ subject to $\sum_{u \in U}z_u \le k$.

The function $P_v(z)$ is shown to be a monotone submodular function~\cite{SomaKIK14}.
Here, a function $f: \{0,1\}^n \rightarrow \mathbb{R}$ is \emph{submodular} if it satisfies
$f(x) + f(y) \ge f(x \vee y) +f(x \wedge y)$
for all $x, y \in \{0,1\}^n$, where $x \vee y$ and $x \wedge y$ denote the vector of component-wise maxima and minima, respectively, i.e., $(x \vee y)_i = \max\{x_i,y_i\}$ and $(x \wedge y)_i = \min\{x_i,y_i\}$.
A function $f$ is \emph{monotone} if it satisfies $f(x) \le f(y)$ for all $x \le y$, i.e., $x_i \le y_i$ for all $i \in [n]$.
Thus the budget allocation problem is a special case of the submodular maximization problem with a cardinality constraint, and it is well-known that the problem is NP-hard~\cite{CornuejolsFN77} and has a $(1-1/e)$-approximation algorithm~\cite{NemhauserWF78}.

\subsection{Stackelberg game}
The \emph{Stackelberg game} is played between two players: the leader and the follower.
Both players can play a mixed strategy, but it is sufficient to consider that the follower plays a pure strategy.
Let $S_L$ and $D_F$ be the sets of pure strategies of the leader and the follower, respectively.
We denote the set of mixed strategies of the leader by $D_L = \{x \in [0,1]^{S_L} \mid \sum_{s \in S_L} x_s = 1\}$, each of which is a probability distribution on pure strategies in $S_L$.
We define $f: D_L \times D_F \rightarrow \mathbb{R}$ and $g: D_L \times D_F \rightarrow \mathbb{R}$ as utility functions of the leader and the follower, respectively.
We define an instance of the game as $\mathcal{G}=(D_L,D_F,f,g)$.
Let $\BR(x) = \argmax\nolimits_{y \in D_F}g(x, y)$ be the set of best responses of the follower against $x$.
In this game, the leader will commit to play a mixed strategy before the follower plays his strategy.
Thus the leader needs to find a mixed strategy $x$ maximizing $f(x,y)$ under the constraint that the follower would choose a best-response pure strategy $y\in \BR(x)$.
More precisely, the goal of this game is 
to find a leader's mixed strategy that forms a strong Stackelberg equilibrium, as indicated below. 

\begin{definition} \label{def:se}
A \emph{strong Stackelberg equilibrium} of $\mathcal{G}$ is a pair $(x^*,y^*)$ that satisfies 
\begin{conference}
$f(x^*,y^*) \ge f(x, y)$
\end{conference}
\begin{fullpaper}
\[f(x^*,y^*) \ge f(x, y)\]
\end{fullpaper}
for all $x \in D_L$, $y \in \BR(x)$, and $y^* \in \BR(x^*)$.
\end{definition}

\section{Stackelberg budget allocation game} \label{sec:model}
In this section, we extend the budget allocation problem with a bipartite influence model to a Stackelberg game. 
For any set $S_L$ and $s \in S_L$,  we denote by $\chi_{s}$ a characteristic vector in $\{0, 1\}^{S_L}$ such that 
$(\chi_s)_{s'} = 1$ for $s' = s$ and $(\chi_s)_{s'} = 0$ for $s'\ne s$ ($s'\in S_L$).
For a mixed strategy $x\in D_L$, the support of $x$ is the set of pure strategies that is played with non-zero probability under $x$,
i.e.,
$\supp(x) = \{s \in S_L \mid x_s > 0\}$.

\subsection{Definition}
Let $G = (U,V;E)$ be a bipartite graph consisting of a set $U$ of $n$ media, a set $V$ of $m$ customers, and a set $E$ of edges between them.
For each $uv \in E$, we denote by $p_{uv}$ a probability that a customer $v$ is activated through a medium $u$ by a leader's or a follower's trial, and by $p_{F,uv}$ a probability that a medium $u$ activates a customer $v$ who has been already activated by the leader.
Two probabilities intuitively mean that $p_{uv}$ is a basic activation probability in the market, and $p_{F,uv}$ is a probability that the follower recaptures customers who were activated by the leader.
Let $k_L$ and $k_F$ be the budgets of the leader and the follower, respectively.
An instance of the \emph{Stackelberg budget allocation game with a bipartite influence model} is parameterized by $\phi = (G= (U,V;E), \{p_{uv}\}_{uv \in E}, \{p_{F,uv}\}_{uv \in E},k_L, k_F)$.

We construct a Stackelberg game $\mathcal{G}$ from an instance $\phi$ as follows.
A pure strategy for the leader (respectively the follower) is a set of at most $k_L$ media (respectively $k_F$ media). 
$D_L$ and $D_F$ of the game $\mathcal{G}$ are defined by setting $S_L=\{z \in \{0,1\}^U \mid \sum_{u \in U}z_u \leq k_L\}$ (or equivalently $S_L=\{S \subseteq U \mid |S|\leq k_L\}$) and $D_F =\{y \in\{0,1\}^U \mid \sum_{u \in U} y_u \le k_F\}$.

Let $v \in V$ be any customer.
Let $z$ and $y$ be a leader's and a follower's pure strategies, respectively. 
The probability that the leader activates $v$ is given by the equation~(\ref{eq:Pvz}).
If $v$ is not activated by the leader, then the activation probability for the follower is given by the same basic probability, that is $P_{v}(y)$.
If $v$ is activated by the leader, then the probability that the follower attracts a customer $v \in V$ away from the leader is $P_{F,v}(y) = 1 - \prod_{u \in N_v:y_u=1} (1-p_{F,uv})$.

\begin{example}
We explain the difference between $p_{uv}$ and $p_{F, uv}$. 
Consider a game instance illustrated in Figure~\ref{fig:e}. 
There are three media $u_1, u_2, u_3$ and four customers $v_1, v_2, v_3, v_4$.
For an arbitrary edge $uv$, $p_{uv} = 0.8$ and $p_{F,uv} = 0.5$.
The budget for the leader and the follower is $k_L=2$ and $k_F=1$, respectively.
At first, the leader plays a mixed strategy $x$ that chooses $\{u_1, u_2\}$ w.p.\ $1$.
Suppose the situation in Figure~\ref{subfig:e_a} where $\{u_1, u_2\}$ is chosen and $v_1$, $v_2$, and $v_3$ are activated w.p.\ 0.8, who are shown in gray. 
After that, the follower plays a pure strategy that chooses $\{u_3\}$. 
In Figure~\ref{subfig:e_b}, the customer $v_2$ switches to the follower w.p.\ $0.5$ if $v_2$ is activated by the leader, and otherwise $v_2$ is activated w.p.\ $0.8$. 
Thus, the probability to activate $v_2$ is $0.96\cdot 0.5+0.04\cdot 0.8=0.512$.  
In addition, $v_4$ is activated w.p.\ $0.8$ because $v_4$ is non-activated. 
\end{example}
\begin{figure}[htbp]
\begin{minipage}[b]{.45\columnwidth}
\centering
\subfigure[The leader's turn]{
\begin{tikzpicture}[xscale=.5,yscale=.4, domain=0:10, cn/.style={circle,fill=white,draw=black,inner sep=2pt}]
\node[cn] at (1,7) (u1) {};
\node[cn] at (1,5) (u2) {};
\node[cn] at (1,3) (u3) {};
\node[cn,fill=lightgray] at (5,8) (v1) {};
\node[cn,fill=lightgray] at (5,6) (v2) {};
\node[cn,fill=lightgray] at (5,4) (v3) {};
\node[cn] at (5,2) (v4) {};

 \draw (1,6) circle [x radius=.5, y radius = 2];

\draw[thick] (u1) -- (v1);
\draw[thick] (u1) -- (v2);
\draw[thick] (u2) -- (v2); 
\draw[thick] (u2) -- (v3);
\draw[thick] (u3) -- (v2);
\draw[thick] (u3) -- (v4);

\draw (u1) node[left,xshift=-5pt] {$u_1$};
\draw (u2) node[left,xshift=-5pt] {$u_2$};
\draw (u3) node[left,xshift=-5pt] {$u_3$};

\draw (v1) node[right,xshift=2pt] {$v_1$};
\draw (v2) node[right,xshift=2pt] {$v_2$};
\draw (v3) node[right,xshift=2pt] {$v_3$};
\draw (v4) node[right,xshift=2pt] {$v_4$};

\end{tikzpicture}\label{subfig:e_a}}
\end{minipage}%
\begin{minipage}[b]{.1\columnwidth}
\mbox{}
\end{minipage}%
\begin{minipage}[b]{.45\columnwidth}
\centering
\subfigure[The follower's turn]{
\begin{tikzpicture}[xscale=.5,yscale=.4, domain=0:10, cn/.style={circle,fill=white,draw=black,inner sep=2pt}]
\node[cn] at (1,7) (u1) {};
\node[cn] at (1,5) (u2) {};
\node[cn,fill=black] at (1,3) (u3) {};
\node[cn,fill=lightgray] at (5,8) (v1) {};
\node[cn,fill=black] at (5,6) (v2) {};
\node[cn,fill=lightgray] at (5,4) (v3) {};
\node[cn,fill=black] at (5,2) (v4) {};

 \draw (1,6) circle [x radius=.5, y radius = 2];

\draw[thick] (u1) -- (v1);
\draw[thick] (u1) -- (v2);
\draw[thick] (u2) -- (v2); 
\draw[thick] (u2) -- (v3);
\draw[thick] (u3) -- (v2);
\draw[thick] (u3) -- (v4);

\draw (u1) node[left,xshift=-5pt] {$u_1$};
\draw (u2) node[left,xshift=-5pt] {$u_2$};
\draw (u3) node[left,xshift=-5pt] {$u_3$};

\draw (v1) node[right,xshift=2pt] {$v_1$};
\draw (v2) node[right,xshift=2pt] {$v_2$};
\draw (v3) node[right,xshift=2pt] {$v_3$};
\draw (v4) node[right,xshift=2pt] {$v_4$};

\end{tikzpicture}\label{subfig:e_b}}
\end{minipage}
\vskip -3mm
\caption{The difference between $p_{uv}$ and $p_{F,uv}$.}\label{fig:e}
\end{figure}
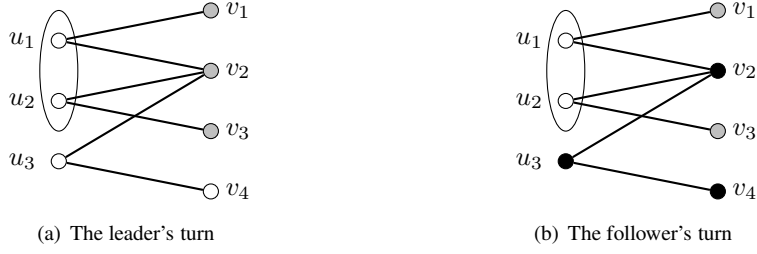

The utility functions $f$ and $g$ 
are given as follows. 
The expected number of customers that are activated by the leader but do not shift to the follower is given by
\begin{align*}
\textstyle
f(z,y)= \sum_{v \in V} P_v(z)(1 - P_{F,v}(y)). 
\end{align*}
The expected number of activated customers for the follower is given by
\begin{align*}
\textstyle
g(z,y)= \sum_{v \in V}&\bigl(P_v(z)P_{F,v}(y)+(1-P_v(z))P_v(y)\bigr).
\end{align*}

When the leader uses a mixed strategy $x$, we abuse the notation and write $P_{v}(x) = \E_{z \sim x}[P_v(z)]$. 
Here, for a probability distribution $x$ over a domain $D$, $z \sim x$ means that we sample $z \in D$ from the distribution $x$.
Similarly, we write $f(x,y)=\E_{z \sim x}[f(z,y)] = \sum_{v \in V} P_v(x)(1 - P_{F,v}(y))$, and $g(x,y)=\E_{z \sim x}[g(z,y)]=\sum_{v \in V}\bigr(P_v(x)P_{F,v}(y)+(1-P_v(x))P_v(y)\bigl)$.

The goal of the Stackelberg budget allocation game with a bipartite influence model is to find a leader's mixed strategy $x$ in a strong Stackelberg equilibrium of the game $\mathcal{G}$.
We define a function $\fBR$ that receives a mixed strategy $x \in D_L$ and returns the leader's utility when the follower takes a best response, i.e., 
\begin{align*}
\fBR(x)=\max\left\{f(x,y)\mid y\in\BR(x)\right\}. 
\end{align*}
We aim to solve
\begin{align}
\max \quad \fBR(x) \quad \quad {\rm s.t.} \quad x\in D_L.\label{eq:problem}
\end{align}
Note that $x$ is an optimal solution to \eqref{eq:problem} if and only if $(x, y)$ is a strong Stackelberg equilibrium, where $y$ is a best response against $x$. 
We can evaluate $\fBR(x)$ for $x\in S_L$
in ${\rm O}(|D_F|\cdot |E| \cdot |\supp(x)|)$ time by evaluating $f(x, y)$ $|D_F|$ times. 
To obtain the value of $f(x, y)$, we evaluate $f(z,y)$ for $z \in \supp(x)$, which takes ${\rm O}(|E|)$ time. 

We now see that the leader's optimal strategy may not be a pure strategy.
\begin{example}
\begin{fullpaper}
Consider an instance with $U=\{u_1,u_2,u_3\}$, $V=\{v_1,v_2,v_3,v_4\}$, 
$E=\{u_1v_1, u_1v_2, u_2v_2, u_2v_3, u_3v_4\}$,
and $k_L=k_F=1$.
For a set $S$ of media, 
$\chi_{S}$ denotes a unit vector in $\{0,1\}^{S_L}$ with $(\chi_{S})_S = 1$. 
The instance including activation probabilities is depicted in Figure~\ref{subfig:exs_a}.
\end{fullpaper}
\begin{conference}
Consider an instance depicted in Figure~\ref{subfig:exs_a} with $k_L=k_F=1$.
\end{conference}
In this case, an optimal strategy for the leader is
$x^*=0.5\chi_{\{u_1\}}+0.5\chi_{\{u_2\}}$
and $\fBR(x^*)=1.1$ where the best response of the follower is $\{u_1\}$.
However, $\fBR(\chi_{\{u_1\}})=\fBR(\chi_{\{u_2\}})=0.6$
and $\fBR(\chi_{\{u_3\}})=0.599$.
\end{example}

We next see that 
the leader may not use the whole budget in her optimal strategy.
\begin{example}\label{ex:budget}
\begin{fullpaper}
Consider an instance depicted in Figure~\ref{subfig:exs_b} where $U=\{u_1,u_2,u_3\}$, $V=\{v_1,v_2\}$, 
$E=\{u_1v_1, u_2v_1, u_2v_2, u_3v_2\}$,
$k_L=3$, and $k_F=1$.
Also, $p_{u_1v_1}=p_{F,u_2v_1}=p_{F,u_2v_2}=p_{u_3v_2}=1$
and $p_{F,u_1v_1}=p_{u_2v_1}=p_{u_2v_2}=p_{F,u_3v_2}=0$.
\end{fullpaper}
\begin{conference}
Consider an instance depicted in Figure~\ref{subfig:exs_b} with $k_L=3$ and $k_F=1$.
\end{conference}
Then $\fBR(\chi_{U})=0$ while $\fBR(\chi_{\{u_1\}})=\fBR(\chi_{\{u_3\}})=1$.
\end{example}

\begin{figure}[htbp]
\vskip -3mm
\begin{minipage}[b]{.5\columnwidth}
\centering
\subfigure[An instance without leader's pure optimal strategies]{
\scalebox{1}{
\begin{tikzpicture}[xscale=.6,yscale=.5, domain=0:10, cn/.style={circle,fill=white,draw=black,inner sep=2pt}]
\draw[draw=white] (-1.3,3) -- (7.3,3);
\node[cn] at (1,7) (u1) {};
\node[cn] at (1,5) (u2) {};
\node[cn] at (1,3) (u3) {};
\node[cn] at (5,8) (v1) {};
\node[cn] at (5,6) (v2) {};
\node[cn] at (5,4) (v3) {};
\node[cn] at (5,3) (v4) {};

\draw[thick] (u1) -- (v1) node[above,pos=.3,sloped,yshift=-2pt] {$(\textcolor{red}{0.1},\,\textcolor{blue}{0})$};
\draw[thick] (u1) -- (v2) node[above,pos=.7,sloped,yshift=-2pt] {$(\textcolor{red}{1},\,\textcolor{blue}{0.5})$};
\draw[thick] (u2) -- (v2) node[below,pos=.7,sloped,yshift=2pt] {$(\textcolor{red}{1},\,\textcolor{blue}{0.5})$};
\draw[thick] (u2) -- (v3) node[below,pos=.3,sloped,yshift=2pt] {$(\textcolor{red}{0.1},\,\textcolor{blue}{0})$};
\draw[thick] (u3) -- (v4) node[below,pos=.4,sloped,yshift=2pt] {$(\textcolor{red}{0.599},\,\textcolor{blue}{0})$};

\draw (u1) node[left,xshift=-2pt] {$u_1$};
\draw (u2) node[left,xshift=-2pt] {$u_2$};
\draw (u3) node[left,xshift=-2pt] {$u_3$};

\draw (v1) node[right,xshift=2pt] {$v_1$};
\draw (v2) node[right,xshift=2pt] {$v_2$};
\draw (v3) node[right,xshift=2pt] {$v_3$};
\draw (v4) node[right,xshift=2pt] {$v_4$};

\end{tikzpicture}}\label{subfig:exs_a}}
\end{minipage}%
\begin{minipage}[b]{.5\columnwidth}
\centering
\subfigure[An instance where the leader should not spend whole her budget]{
\scalebox{1}{
\begin{tikzpicture}[xscale=.6,yscale=.5, domain=0:10, cn/.style={circle,fill=white,draw=black,inner sep=2pt}]
\draw[draw=white] (2.5,3) -- (9.5,3);

\node[cn] at (4,7.5) (u1) {};
\node[cn] at (4,5) (u2) {};
\node[cn] at (4,2.5) (u3) {};
\node[cn] at (8,6) (v1) {};
\node[cn] at (8,4) (v2) {};

\draw[thick] (u1) -- (v1) node[below,pos=.3,sloped,yshift=2pt] {$(\textcolor{red}{1},\,\textcolor{blue}{0})$};
\draw[thick] (u2) -- (v1) node[below,pos=.7,sloped,yshift=2pt] {$(\textcolor{red}{0},\,\textcolor{blue}{1})$};
\draw[thick] (u2) -- (v2) node[below,pos=.3,sloped,yshift=2pt] {$(\textcolor{red}{0},\,\textcolor{blue}{1})$};
\draw[thick] (u3) -- (v2) node[below,pos=.7,sloped,yshift=2pt] {$(\textcolor{red}{1},\,\textcolor{blue}{0})$};

\draw (u1) node[left,xshift=-2pt] {$u_1$};
\draw (u2) node[left,xshift=-2pt] {$u_2$};
\draw (u3) node[left,xshift=-2pt] {$u_3$};

\draw (v1) node[right,xshift=2pt] {$v_1$};
\draw (v2) node[right,xshift=2pt] {$v_2$};

\end{tikzpicture}}\label{subfig:exs_b}}
\end{minipage}
\vskip -3mm
\caption{Examples of the Stackelberg budget allocation game, where the pair of numbers on each edge $uv$ represents the activation probabilities $p_{uv}$ and $p_{F,uv}$.}\label{fig:exs}
\end{figure}
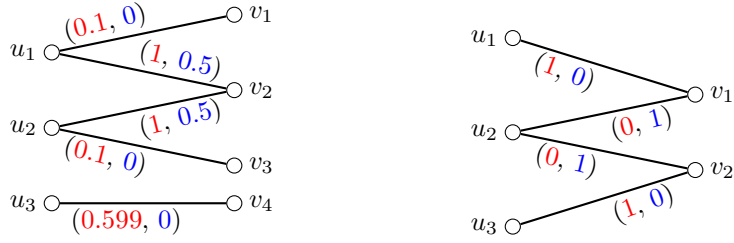

\begin{conference}
There also exists an instance without a pure Stackelberg equilibrium (see Example~4.4 in the upcoming full version).
\end{conference}
\begin{fullpaper}
There also exists an instance without a pure Stackelberg equilibrium. 
\begin{example}\label{ex:bilinear}
Consider an instance with $U=\{u_1,u_2,u_3,u_4\}$
and $V=V_1\cup V_2\cup V_3\cup V_4$
where $V_1=\{v_{1,1},\dots,v_{1,10}\}$ and $V_i=\{v_{i,1},\dots,v_{i,6}\}$ $(i=2,3,4)$.
Let $E=\{uv\mid u=u_i, v\in V_i, i=1,2,3,4\}$
and let $p_{uv}=1$ and $p_{F,uv}=0.5$ for all $uv\in E$.
We have $|N_v| = 1$ for all $v \in V$.
There exists a mixed Stackelberg equilibrium $(x,y)$, where 
$x=\frac{1}{3}(\chi_{\{u_1,u_4\}}+\chi_{\{u_1,u_2,u_4\}}+\chi_{\{u_1,u_3,u_4\}})$ and $y=\chi_{\{u_2,u_3\}}$.
However, there is no pure Stackelberg equlibrium in this instance.
\end{example}
\end{fullpaper}

\subsection{Hardness}
In this subsection, we show hardness results. 
We observe that finding a leader's optimal pure strategy when $k_F=0$ is equivalent to the optimal budget allocation problem.
Thus, it is NP-hard to find the leader's mixed strategy that forms a Stackelberg equilibrium even if $k_F=0$, since our problem~\eqref{eq:problem} when $k_F=0$ always has the leader's optimal strategy that is pure.
It is also known that the approximation ratio $1-1/e$ is best possible for the maximum coverage problem under the assumption that P $\neq$ NP~\cite{Feige1998}. 
Hence, our problem \eqref{eq:problem} is also inapproximable within ratio $1-1/e$ unless P $=$ NP. 

Moreover, when $k_F$ is not a fixed constant, it is even NP-hard to evaluate $\fBR(x)$ for a given $x \in D_L$. 
The proof is reducing from the maximum coverage problem, which is shown to be NP-hard (see e.g.,~\cite{Hochbaum96}).
Given an integer $k$ and a collection of sets $\mathcal{S}=\{S_{1},S_{2},\ldots ,S_{n}\}$, the \emph{maximum coverage problem} is to find a subset $\mathcal{S}'\subseteq \mathcal{S}$ of at most $k$ sets such that the number of covered elements $\left|\bigcup _{S_{i}\in \mathcal{S}'}{S_{i}}\right|$ is maximized.
\begin{conference}
See the upcoming full version for the proof.
\end{conference}

\begin{theorem} \label{thm:f_maxcov}
It is NP-hard to compute $\fBR(x)$ for $x \in D_L$. 
\end{theorem}
\begin{fullpaper}
\begin{proof}
Let $(k, \mathcal{S})$ be any instance of the maximum coverage problem. 
We consider an instance of Stackelberg budget allocation problem: $U=\{S_1,\dots,S_n\}$, $V=\bigcup_{i=1}^n S_i$, $E=\{S_i v \mid v\in S_i,~i=1,\dots,n\}$,
and $p_{uv}=p_{F,uv}=1$ for each $uv\in E$.
Let $k_L=n$ and $k_F=k$.

We fix a leader's mixed strategy as $x = \chi_{U}$. 
We denote $z$ be an all-one-vector. 
To evaluate $\fBR(x)$, it is necessary to know $\max_{y \in D_F}g(x, y)$
since $\fBR(x)=|V|-\max_{y\in D_F}g(x,y)$. 

We show that there exists $\mathcal{S}'$ that covers at least $\alpha$ elements if and only if $g(x,y) \geq \alpha$ for some $y \in D_F$.
By construction, we have $g(x,y)=\sum_{v \in V}P_{F, v}(y)$. 
In addition, $P_{F,v}(y) = 1$ if $v$ is covered by some set $S$ with $y_{S} = 1$, and $P_{F,v}(y) = 0$ otherwise. 
Thus, $g(x, y) = \left|\bigcup _{S \in \supp(y)}{S}\right|$. 

Then, if $\mathcal{S}' \subseteq \mathcal{S}$ that covers at least $\alpha$ elements, then $y\in D_F$ attains $g(x,y) \geq \alpha$, where $y$ is defined by $y_{S} = 1$ if $S \in \mathcal{S}'$ and $y_{S}=0$ otherwise. 
Conversely, if $y$ satisfies $g(x,y) \geq \alpha$, then $\supp(y)$ covers at least $\alpha$ elements. 
This completes the proof. 
\end{proof}
\end{fullpaper}

\section{Algorithms for non-disjoint customers} \label{sec:alg}

In this section, for the non-disjoint customers setting, which has no assumption about the graph structure, we propose two types of algorithms for \eqref{eq:problem}. 
Let $\mathcal{G}$ be a game instance created from an instance  $\phi=(G= (U,V;E), \{p_{uv}\}_{uv \in E}, \{p_{F,uv}\}_{uv \in E},\allowbreak k_L, k_F)$ and let $\Lambda$ be its data size. 
Due to the hardness result~(Theorem \ref{thm:f_maxcov}), 
in this section we assume that $k_F$ is a constant.

\subsection{Approximation algorithm via zero-sum game}\label{subsec:approx}

We shall approximately solve a game $\mathcal{G}$ by solving a zero-sum game close to $\mathcal{G}$.
The core idea of constructing such a zero-sum game is to keep the same set of best-responses of the follower for any strategy of the leader as $\mathcal{G}$. 
Let us focus on the structure of $f$ and $g$, which include the term $-\sum_{v\in V}P_v(x)P_{F,v}(y)$ and its negation, respectively.
We define a utility function for the leader as 
\begin{align*}
\textstyle
    \Phi(x,y) 
    &\textstyle= -g(x,y)+\sum_{v \in V}P_v(x) \\ 
    &\textstyle= \sum_{v\in V}[P_v(x)(1-P_{F,v}(y))-P_v(y)(1-P_v(x))].
\end{align*}
Note that $C\coloneqq \max_{y\in D_F}\sum_{v\in V}P_v(y)\ge -\Phi(x,y)$
and we can compute $C$ in polynomial time since $|D_F|$ is polynomially bounded.
Let $\mathcal{G}_{\Phi}$ be a zero-sum game $(D_L,D_F,\Phi,\allowbreak-\Phi)$.

For reals $\alpha \in [0,1]$ and $\epsilon \ge 0$, we call an algorithm \emph{$(\alpha,\epsilon)$-approximation} for 
$\mathcal{G}$ (resp.\ $\mathcal{G}_{\Phi}$) if it provides a strategy profile $(x',y')$ such that $y'\in \BR(x')$ and 
$f(x',y') \ge \alpha\cdot \max_{x\in D_L}\fBR(x)-\epsilon$ (resp.\ $\Phi(x',y') \ge \alpha\cdot \max_{x\in D_L, y\in\BR(x)}\Phi(x,y)-\epsilon$). 
Such $(x',y')$ is called an \emph{$(\alpha,\epsilon)$-approximate solution}.

\begin{lemma}\label{lem:zero}
Let $(x',y')$ be an $(\alpha,\epsilon)$-approximate solution of a zero-sum game $\mathcal{G}_{\Phi}$, 
and let $(x^*,y^*)$ be a strong Stackelberg equilibrium of the original game $\mathcal{G}$. 
Let $\epsilon_1 \coloneqq \sum_{v \in V}(1-P_v(x'))P_v(y')$ and $\epsilon_2 \coloneqq \sum_{v \in V}(1-P_v(x^*))P_v(y^*)$.
Then $(x',y')$ is an $(\alpha,\alpha\epsilon_2-\epsilon_1+\epsilon)$-approximate solution for the game $\mathcal{G}$.
\end{lemma}
\begin{proof}
We remark that $f(x,y)$ can be rewritten by $\Phi(x,y)$ as $f(x,y) = \Phi(x,y)+\sum_{v\in V}(1-P_v(x))P_{v}(y)$.
Let $(\tilde{x},\tilde{y})$ be the minimax strategy of $\mathcal{G}_{\Phi}$.
We have 
\begin{align*}
    f(x',y') & = \Phi(x',y') + \epsilon_1
        \ge \alpha \Phi(\tilde{x},\tilde{y}) - \epsilon + \epsilon_1\\
        & \ge \alpha \Phi(x^*,y^*) - \epsilon + \epsilon_1 
         = \alpha f(x^*,y^*) - (\alpha \epsilon_2 - \epsilon_1 + \epsilon),
\end{align*}
where the second inequality holds by $\Phi(\tilde{x},\tilde{y}) =\max_{x\in D_L, y \in \BR(x)} \Phi(x, y) \ge \Phi(x^*,y^*)$.
\end{proof}

To find an approximate strong Stackelberg equilibrium, it suffices to find an approximate minimax strategy for $\mathcal{G}_{\Phi}$.
Note that since 
$|S_L|$ is an exponential size, finding a minimax strategy for $\mathcal{G}_{\Phi}$ is still intractable.

To this end, we use the \emph{multiplicative weight update method}~\cite{AroraHK12}.
Based on this method, Kawase and Sumita~\cite{KawaseS19} showed that,
for any nonnegative monotone submodular functions $h_1,\dots,h_\nu\colon \{0,1\}^{n}\to\mathbb{R}_+$
and $\epsilon>0$,
there exists an algorithm 
that finds a $(1-1/e-\epsilon)$-approximate solution of $\max_{x\in D_L}\min_{i\in[\nu]}\E_{s\sim x}[h_i(s)]$ in polynomial time in $n$, $\nu$ and $1/\epsilon$.
We set $h_y(z) = \Phi(z,y)+C$ for all pure strategies $z \in S_L$ and $y \in D_F$. 
By the definition, $h_y$ is nonnegative monotone submodular for any $y \in D_F$. 
Thus, we see that we can compute a $(1-1/e-\epsilon)$-approximate solution for $\max_{x\in D_L}\min_{y\in D_F}(\Phi(x,y)+C)$ in polynomial time in $\Lambda$ and $1/\epsilon$. 
This solution is $\bigl(1-1/e-\epsilon,(1/e+\epsilon)C\bigr)$-approximate for $\mathcal{G}_{\Phi}$.
Therefore, by Lemma \ref{lem:zero}, we observe the following result.

\begin{theorem}
For any $\epsilon>0$,
there exists a $(1-1/e-\epsilon,\beta)$-approximation algorithm
where $\beta=(1-1/e)\epsilon_2-\epsilon_1+(1/e+\epsilon)C$
and the running time is polynomial with respect to $\Lambda$ and $1/\epsilon$.
\end{theorem}

\subsection{Heuristic algorithm}
In this subsection, we propose a heuristic algorithm. 
Intuitively, in the algorithm, the players fictitiously play a game $\ell$ times.
Here $\ell$ is a parameter. 
Let us assume that the leader would know that the follower estimates the leader's mixed strategy by observing the past budget allocations.
In every phase, the leader needs to allocate her budgets so that the mixed strategy estimated by the follower maximizes the leader's utility.
The algorithm outputs a mixed strategy by repeating this phase $\ell$ times. 

We describe informally our algorithm, which is summarized in Algorithm \ref{alg:heu}. 
The algorithm repeatedly computes $\ell$ pure strategies $\chi_{S_1}, \ldots, \chi_{S_\ell} \in D_L$, and outputs the best mixed strategy among $\frac{1}{i}(\chi_{S_1}+\cdots + \chi_{S_i})$ $(i=1,\dots,\ell)$. 
At first round, $\chi_{S_1}$ is chosen to maximize $\fBR(x)$. 
Each $\chi_{S_i}$ is computed greedily (lines 4--9). 

\begin{algorithm}[t]
\caption{Proposed heuristic}\label{alg:heu}
\SetKwInOut{Input}{input}\Input{a parameter $\ell\in\mathbb{Z}_+$\quad\textbf{output:} a mixed strategy}
$x,x^*\leftarrow \chi_{\emptyset}$\;
\For{$i=1$ to $\ell$}{
    $S \leftarrow \emptyset$\;
    \For{$j = 1$ to $k_L$}{
        $x^u\leftarrow \frac{i-1}{i}x+\frac{1}{i}\chi_{S\cup\{u\}}$ $(u\in U)$\;
        $r \in \argmax\nolimits_{u \in U\setminus S} \fBR(x^u)$\;
        \lIf{$\fBR(x^r)\ge \fBR(x)$}{
            $S \leftarrow S\cup\{r\}$
        }
        \lElse{break}
    }
    $x \leftarrow \frac{i-1}{i}x+\frac{1}{i}\chi_{S}$\;
    \lIf{$\fBR(x^*)<\fBR(x)$}{
        $x^*\leftarrow x$
    }
}
\Return $x^*$\;
\end{algorithm}

In each round $i$, we evaluate $\fBR$ ${\rm O}(n\cdot k_L)$ times, and each evaluation of $\fBR$ takes ${\rm O}(|D_F|\cdot |E|\cdot i)$ time. 
Thus the total running time is ${\rm O}(|D_F|\cdot |E|\cdot n\cdot k_L\cdot \ell^2)$.

\section{Algorithm for disjoint customers} \label{sec:ana}
In this section, we focus on the disjoint customers setting where each customer is interested in only one medium, i.e., $|N_v| = 1$ for all $v \in V$. 
This means that the utility functions $f,g$ are bilinear.
In this special case, we propose an LP-based algorithm, and modify it so that it runs fast when $|D_F|$ is small.
We denote by $\Lambda$ the data size of an input game instance $(G, \{p_{uv}\}_{uv \in E}, \{p_{F,uv}\}_{uv \in E}, k_L, k_F)$. 
The following proposition is the main result in this section. 
\begin{proposition}\label{prop:bilinear}
When $|N_v| = 1$ for all $v \in V$, we can find a strong Stackelberg equilibrium $(x, y)$ in polynomial time with respect to $|D_F|$ and $\Lambda$. 
\end{proposition}

As we will see in Section~\ref{alg:mlp}, it is easy to compute a strong Stackelberg equilibrium by a multiple LP formulation. 
The running time is polynomial with respect to $\lambda$, $|S_L|$, and $|D_F|$. 
However, this is not sufficient since $|S_L|$ could be exponentially large with respect to $\lambda$ and $|D_F|$.
To remove the dependency on $|S_L|$, we reduce the size of each LP in Section~\ref{sec:agg}. 
The idea is a projection of a leader's mixed strategy $x \in [0,1]^{S_L}$ onto a fractional budget allocation $r \in [0,1]^U$.

\subsection{Multiple LP formulation} \label{alg:mlp}
We first describe a simple exact algorithm to solve \eqref{eq:problem}. 
The problem \eqref{eq:problem} is rewritten as 
\begin{align}\label{eq:problem alt}
\begin{array}{rl}
\max &  f(x, y) \\
{\rm s.t.} & x \in D_L, \\
&g(x, y) \geq g(x, y') \quad \forall y' \in D_F. 
\end{array}
\end{align}
 When we fix $y=y^*$, LP~\eqref{eq:problem alt} is equivalent to the following LP:
\begin{align}
\!\!\!\!\!\!
\begin{array}{rl}
\max&\sum_{z\in S_L} f(z,y^*)x_z\\
\text{s.t.}
&\sum_{z \in S_L}(g(z,y^*)-g(z,y')) x_z \geq 0 \ \ \  \forall y'\in D_F, \\
&\sum_{z\in S_L} x_z=1,\\
&x_z\ge 0 \quad \forall z \in S_L.
\end{array}\label{eq:LP_nonlinear}
\end{align}
The simple algorithm solves \eqref{eq:problem alt} exactly by solving \eqref{eq:LP_nonlinear} for each $y^* \in D_F$. 
Each LP \eqref{eq:LP_nonlinear} is solvable in polynomial time with respect to $\Lambda$, $|S_L|$ and $|D_F|$, and the algorithm produces $|D_F|$ instances of LP \eqref{eq:LP_nonlinear}.
Thus this algorithm runs in polynomial time with respect to $\Lambda$, $|S_L|$, and $|D_F|$.

\subsection{Reduced formulation}\label{sec:agg}
Let $A$ be a matrix in $\{0,1\}^{U \times S_L}$ whose rows are all pure strategies. 
For notational convenience, we denote $p'_{uv}=p_{uv} - p_{F,uv}$ for each $uv \in E$. 
We denote by a fractional budget allocation $r \in [0,1]^U$ with $\sum_{u \in U} r_u \leq k_L$.
We remark that a fractional budget allocation is a different notion from a mixed strategy $x \in D_L$; the former is uniquely defined from the latter as $r_u = \sum_{S: u\in S} x_S \ (u\in U)$, but the converse may not hold.

We first observe that $A$ projects a mixed strategy $x$ to a fractional budget allocation $Ax \in [0,1]^U$.
Let $Q=\{r \in[0,1]^U \mid r=Ax, x \in D_L\}$.
%
%
\begin{lemma}\label{lemma:z in Q}
For any vector $z$, it holds that $z \in Q$ if and only if 
\begin{align}\label{eq:matroid poly_uniform}
\textstyle
0 \leq z \leq 1, \ \sum_{u \in U} z_u \leq k_L.
\end{align}
\end{lemma}
\begin{fullpaper}
\begin{proof}
For any set $S$ of media, we define $q(S) = \min\{ |S|, k_L \}$. 
It is not difficult to see that a vector $z$ is in $Q$ if and only if $z$ satisfies
\begin{align}\label{eq:matroid poly}
\textstyle
\sum_{u \in S} z_u \leq q(S) \ (\forall S \subseteq U),  \quad z \geq 0.
\end{align}
Moreover, $z$ satisfies \eqref{eq:matroid poly} if and only if 
it satisfies \eqref{eq:matroid poly_uniform}. 
The only-if part is clear. 
To see the if part, assume that $\sum_{u \in S} z_u \leq q(S)$ for some $S$. 
If $|S| \geq k_L$, then it holds that $1^\top z \geq \sum_{u \in S} z_u > q(S)=k_L$.  
Otherwise, i.e., $|S| < k_L$, since $\sum_{u \in S} z_u > |S|$, we have $z_u > 1$ for some $u \in S$. 
\end{proof}
\end{fullpaper}

We can rewrite $P_v$ and $P_{F,v}$ as $P_v(z) = p_{uv}z_{u}$ and $P_{F,v}(y) = p_{F,uv}y_{u}$, where $u$ is the only neighbor of $v$. 
Then $f$ and $g$ are simplified as
\begin{align*}
f(x, y) &=\textstyle\sum_{u \in U} \sum_{v \in N_u} p_{uv} (Ax)_u (1 - p_{F,uv}y_{u}),\\
g(x, y) &=\textstyle\sum_{u \in U} \sum_{v \in N_u} p_{uv} y_u\left(1-p'_{uv} (Ax)_u \right).
\end{align*}
The utility functions $f(x, y)$ and $g(x, y)$ are bilinear. 
Moreover, they depend on a fractional budget allocation $Ax \in [0,1]^U$ rather than $x$. 
\begin{lemma}\label{lemma:bilinear}
Assume that $|N_v|=1$ for all $v \in V$. 
For each $x \in D_L$ and $y \in D_F$, it holds that $f(x, y) = f(x', y)$ and $g(x, y)=g(x',y)$ for any $x \in D_L$ such that $Ax = Ax'$. 
\end{lemma}

This lemma gives us an intuition that we solve \eqref{eq:LP_nonlinear} for a fractional budget allocation $r$ and recover a mixed strategy $x$. 
We claim that LP \eqref{eq:LP_nonlinear} is polynomially equivalent to
\begin{align}
\begin{array}{rl}
\max&\ \sum_{u \in U} \sum_{v \in N_u} p_{uv}(1 - p_{F,uv}y^*_{u})r_u  \\
\text{s.t.}&\ \sum_{u \in U} \sum_{v \in N_u} p_{uv} y'_u r'_u \geq 0, \\
			  &\  y'_u = y^*_u - y_u \quad \forall u \in U,  y\in D_F, \\
           &\ r'_u = 1-p'_{uv}r_{u} \quad \forall u \in U, \\
           &\ \sum_{u\in U}r_u\le k_L, \\
           &\ r_u \in [0,1] \quad \forall u \in U. 
\end{array}\label{eq:LP_linear}
\end{align}
Indeed, if $(x, y)$ is an optimal solution for \eqref{eq:LP_nonlinear}, then we obtain an optimal solution $(r, y)$ for \eqref{eq:LP_linear} by setting $r=Ax$. 
Conversely, let $(r, y)$ be any optimal solution for \eqref{eq:LP_linear}. 
We observe that $r \in Q$ by Lemma~\ref{lemma:z in Q}.
If we can construct a mixed strategy $x \in D_L$ such that $r=Ax$, then we see that $(x, y)$ is an optimal solution for \eqref{eq:LP_nonlinear} by Lemma \ref{lemma:bilinear}. 
In the following, we show that 
we can recover $x \in D_L$ such that $r=Ax$ in polynomial time with respect to $\Lambda$. 
%
%
\begin{lemma}\label{lemma:recover}
For any $r^* \in Q$, there exists a polynomial-time algorithm that finds a mixed strategy $x \in D_L$ such that $|\supp(x)| \leq n+1$ and $r^* = \sum_{z \in \supp(x)} x_z z$. 
\end{lemma}
Note that the mixed strategy $x$ in the statement always exists by Carath\'eodory's theorem. 
Lemma \ref{lemma:recover} holds even if some $|N_v|$ is not necessarily equal to one.
A leader's strategy in a Stackelberg equilibrium may have the support of a large size. 

\begin{fullpaper}
To show Lemma \ref{lemma:recover}, we use a result by Gr\"otschel et al.~\cite{GroetschelLS2012}. 
The \emph{separation problem for $Q$} is the problem that receives a vector $r^*$ and either asserts $r^* \in Q$ or finds a hyperplane $a^\top r = b$ such that $a^\top r^* >b$ and $a^\top r \leq b \ (\forall r \in Q)$.
\begin{theorem}[Gr\"otschel et al.~\cite{GroetschelLS2012}]\label{thm:Groetschel}
If the separation problem for $Q$ can be solved in polynomial time, then there is a polynomial time algorithm that, on any input $r^* \in Q$, 
computes $n+1$ pure strategies $z^1, \ldots, z^{n+1}$ and coefficients $\lambda_1, \ldots,  \lambda_{n+1} \geq 0$ such that $r^* = \sum_{i=1}^{n+1} \lambda_i z^i$ and $\sum_{i=1}^{n+1} \lambda_i = 1$. 
\end{theorem}
\begin{proof}[Proof of Lemma \ref{lemma:recover}]
By Lemma~\ref{lemma:z in Q}, the separation problem for $Q$ can be solved in polynomial time as follows. 
We check the feasibility of the $2n+1$ inequalities in \eqref{eq:matroid poly_uniform}. 
If yes, then we assert $z^* \in Q$; otherwise, output any inequality that is not satisfied by $z^*$. 

Then by Theorem~\ref{thm:Groetschel}, we obtain a mixed strategy $x$ by setting $x_{z^i} = \lambda_i  \ (i=1, \ldots, n+1)$ and other elements are zero. 
This strategy $x$ needs memory of polynomial size with respect to $\Lambda$. 
\end{proof}
\end{fullpaper}

Therefore, we can solve \eqref{eq:problem alt} by solving \eqref{eq:LP_linear} and recovering a mixed strategy $x \in D_L$ for each $y^* \in D_F$. 
This algorithm generates $|D_F|$ instances of LP \eqref{eq:LP_linear} and each instance can be solved in polynomial time in $\Lambda$ and $|D_F|$. 
Note that the data size of the LP~\eqref{eq:LP_linear} is bounded by polynomial in $\Lambda$ and $|D_F|$. 
The recovered mixed strategy $x$ has polynomial size in $\Lambda$. 
By summarizing the above arguments, Proposition~\ref{prop:bilinear} is proved. 

\section{Experiments} \label{sec:exp}

In this section, we evaluate the performance of
the proposed approximation algorithm
and the heuristic algorithm on real-world datasets. 
We execute the approximation algorithm \textsf{\small{Approx}}
(the algorithm based on MWU described in Section~\ref{subsec:approx} with $100$ iterations and $\epsilon = 0.5$),
and the heuristic algorithm \textsf{\small{Prop.}} (Algorithm~\ref{alg:heu} with $\ell =10$).
We compare the above algorithms with a baseline algorithm \textsf{\small{Greedy}}, which greedily chooses $k_L$ media to maximize $\sum_{v \in V} P_{v}(z)$.
We conduct a series of experiments on Movielens~\cite{MovieLens} and Yahoo! webscope~\cite{Yahoo2007} datasets to examine the leader's utility.
The dataset MovieLens is constructed from MovieLens 100K 
Dataset\footnote{\url{http://grouplens.org/datasets/movielens/100k/}} with 100,000 ratings ($1$ to $5$) to 1,700 movies by 1,000 users.
From the dataset, we select top $n$ frequently rated movies and constructed a bipartite graph $G$ with $n = 20$ media (movies) and $m = 844$ customers (users) with $|E| = 3506$ edges.
The dataset Yahoo! Webscope is constructed from Yahoo! Search Marketing Advertiser Bidding Data\footnote{\url{https://webscope.sandbox.yahoo.com/catalog.php?datatype=a}}, which contains a bipartite graph between 1,000 search keywords and 10,475 accounts, where each edge represents one bid to advertisement on the keyword with the bid price.
From the dataset, we select top $n$ frequently bidden keywords and constructed a bipartite graph $G$, which has $n = 50$ media (keywords) and $m = 447$ customers (accounts) with $|E| = 871$ edges.
$\U(a, b)$ denotes an uniform distribution with maximum and minimum values $a$ and $b$.
For the above bipartite graphs,
we set each basic activation probability as $ p_{uv} \in \U(0, 0.2)$ for $uv \in E $ as in Wilder
and Vorobeychik~\cite{WilderV19}.
We generate two types of instances;
for each edge $uv \in E$,  the activation probability $p_{F,uv}$ is drawn from a distribution ${\cal D_F}=\U(0, 0.2)$ in the first type of instances,
whereas
that is drawn from a distribution ${\cal D_F}=\U(0.1, 0.9)$ in the second type of instances
that models a scenario
where the follower aims to take customers away from the leader.
We set the leader's budget as $k_L=1, 2, 4$,
whereas the follower's budget is set to be $k_F=2$.
The results reported in Table~\ref{table:result_real} indicate that
our algorithms clearly outperform \textsf{\small{Greedy}}
especially when the follower is eager to strip the leader of her customers; that is, when ${\cal D_F}=\U(0.1, 0.9)$.

\newcommand{\s}{\phantom{.0}}
\begin{conference}
\begin{table}[ht]
\centering
\caption{Results averaged over 30 instances for real-world datasets. }\label{table:result_real}
\vskip -5pt
\setlength{\tabcolsep}{1pt}
\renewcommand{\arraystretch}{.8}
\begin{minipage}{.5\linewidth}
\begin{flushleft}
\begin{tabular}{c|c|rrr}
\multicolumn{5}{c}{\small MovieLens {\scriptsize $((n,m,|E|)=(20, 844, 3506 ))$}}\\\toprule
${\cal D_F}$ & ($k_L, k_F$)       & Greedy               &Approx &Prop.\\\midrule
$\U(\s0, 0.2)$ &$(1, 2)$ & 37.05 & 37.05 & 37.05\\
$\U(\s0, 0.2)$ &$(2, 2)$ & 65.10 & 65.10 & 65.24\\
$\U(\s0, 0.2)$ &$(4, 2)$ &114.22 &114.22 &114.22\\
$\U(0.1, 0.9)$ &$(1, 2)$ & 14.34 & 14.55 & 17.22\\
$\U(0.1, 0.9)$ &$(2, 2)$ & 24.22 & 29.97 & 31.87\\
$\U(0.1, 0.9)$ &$(4, 2)$ & 54.46 & 54.46 & 56.12\\\bottomrule
\end{tabular}
\end{flushleft}
\end{minipage}%
\begin{minipage}{.5\linewidth}
\begin{flushright}
\begin{tabular}{c|c|rrr}
\multicolumn{5}{c}{\small Yahoo! Webscope {\scriptsize$((n,m,|E|)=(50, 447, 871 ))$}}\\\toprule
${\cal D_F}$ & ($k_L, k_F$) & Greedy &Approx &Prop.\\\midrule
$\U(\s0, 0.2)$ &$(1, 2)$ & 5.42 & 5.42 & 5.46\\
$\U(\s0, 0.2)$ &$(2, 2)$ &10.68 &10.68 &10.72\\
$\U(\s0, 0.2)$ &$(4, 2)$ &19.56 &19.56 &19.55\\
$\U(0.1, 0.9)$ &$(1, 2)$ & 2.20 & 3.00 & 3.67\\
$\U(0.1, 0.9)$ &$(2, 2)$ & 5.03 & 6.36 & 7.05\\
$\U(0.1, 0.9)$ &$(4, 2)$ &11.72 &12.68 &13.31\\\bottomrule
\end{tabular}
\end{flushright}
\end{minipage}
\end{table}
\end{conference}

\begin{fullpaper}
\begin{table}[htb]
\centering
\caption{Results averaged over 30 instances for real-world datasets. }\label{table:result_real}
\setlength{\tabcolsep}{5pt}
\begin{tabular}{c|c|rrr}
\multicolumn{5}{c}{\small MovieLens {\scriptsize $((n,m,|E|)=(20, 844, 3506 ))$}}\\\toprule
${\cal D_F}$ & ($k_L, k_F$)       & Greedy               &Approx &Prop.\\\midrule
$\U(\s0, 0.2)$ &$(1, 2)$ & 37.05 & 37.05 & 37.05\\
$\U(\s0, 0.2)$ &$(2, 2)$ & 65.10 & 65.10 & 65.24\\
$\U(\s0, 0.2)$ &$(4, 2)$ &114.22 &114.22 &114.22\\
$\U(0.1, 0.9)$ &$(1, 2)$ & 14.34 & 14.55 & 17.22\\
$\U(0.1, 0.9)$ &$(2, 2)$ & 24.22 & 29.97 & 31.87\\
$\U(0.1, 0.9)$ &$(4, 2)$ & 54.46 & 54.46 & 56.12\\\bottomrule
\multicolumn{5}{c}{} \\
\multicolumn{5}{c}{\small Yahoo! Webscope {\scriptsize$((n,m,|E|)=(50, 447, 871 ))$}}\\\toprule
${\cal D_F}$ & ($k_L, k_F$) & Greedy &Approx &Prop.\\\midrule
$\U(\s0, 0.2)$ &$(1, 2)$ & 5.42 & 5.42 & 5.46\\
$\U(\s0, 0.2)$ &$(2, 2)$ &10.68 &10.68 &10.72\\
$\U(\s0, 0.2)$ &$(4, 2)$ &19.56 &19.56 &19.55\\
$\U(0.1, 0.9)$ &$(1, 2)$ & 2.20 & 3.00 & 3.67\\
$\U(0.1, 0.9)$ &$(2, 2)$ & 5.03 & 6.36 & 7.05\\
$\U(0.1, 0.9)$ &$(4, 2)$ &11.72 &12.68 &13.31\\\bottomrule
\end{tabular}
\end{table}
\end{fullpaper}

\section{Conclusion} \label{sec:con}
We formalized a new model called the \emph{Stackelberg budget allocation game with a bipartite influence model}.
For the general case of our model, we proposed two algorithms: an approximation algorithm which has provable guarantee and a heuristic algorithm empirically outputs a better solution.
We remark that, to the best of our knowledge, our approximation algorithm is the first algorithm with a provable guarantee for the non-zero sum submodular Stackelberg game. 
When the utility functions are bilinear, we proposed our LP-based algorithm and showed that it runs in polynomial time when the follower's budget is constant. 
We remark that in this case, we can generalize the budget constraint to a \emph{matroid} constraint and show a similar result. 
Finally, experimental results indicate that our approximation and heuristic algorithms empirically output good quality solutions especially in the setting that the follower is a powerful competitor.

\section*{Acknowledgments}
This work was partially supported by JST ERATO Grant Number JPMJER1201, Japan, and JSPS KAKENHI Grant Numbers JP17K12744, JP18J23034, JP16K16005, JP17K12646, JP17K00028 and JP18H05291, Japan.

\bibliographystyle{plainurl}
\small
\bibliography{aaai}

\end{document}